\documentclass[12pt,reqno]{amsart}
\usepackage{amsthm,amsfonts,amssymb,euscript}
\usepackage{xcolor}

\def\emph#1{{\it #1}}
\def\textbf#1{{\bf #1}}

\newcommand{\bea}{\begin{eqnarray}}
\newcommand{\eea}{\end{eqnarray}}
\def\beaa{\begin{eqnarray*}}
\def\eeaa{\end{eqnarray*}}

\def\ub{{\underline{u}}}

\def\ba{\begin{array}}
\def\ea{\end{array}}
\def\be#1{\begin{equation} \label{#1}}
\def \eeq{\end{equation}}

\newcommand{\nn}{\nonumber}

\def\nn{\nonumber}

\def\tr{\mbox{tr}}

\renewcommand{\div}{\mbox{div }}

\def\a{\alpha}
\def\alphab{{\underline\alpha}}
\def\b{\beta}
\def\betab{{\underline\beta}}

\def\aa{\alphab}
\def\bb{\betab}

\def\ga{\gamma}
\def\Ga{\Gamma}
\def\de{\delta}
\def\De{\Delta}
\def\ep{\epsilon}

\def\la{\lambda}

\def\si{\sigma}

\def\om{\omega}
\def\omegab{{\underline\omega}}
\def\Om{\Omega}

\def\th{\theta}

\def\ze{\zeta}
\def\nab{\nabla}

\def\bb{\underline{\b}}
\def\lap{\Delta}

\newcommand{\trchb}{\tr \chib}

\newcommand{\chih}{\hat{\chi}}
\newcommand{\chib}{\underline{\chi}}

\newcommand{\etab}{\underline{\eta}}
\newcommand{\chibh}{\underline{\hat{\chi}}\,}
\newcommand{\omb}{{\underline{\om}}}
\def\f14{\frac{1}{4}}
\def\f12{{\frac{1}{2}}}

\def\c{\cdot}

\newcommand{\Ls}{{\mathcal L} \mkern-10mu /\,}

\newcommand{\DD}{{\mathcal D}}

\def\ub{\underline{u}}

\def\Lb{{\underline{L}}}

\def\Hb{{\underline{ H}}}
\def\pr{\partial}

\def\chih{\hat{\chi}}
\def\trch{\mbox{tr}\chi}

\usepackage{graphics}
\usepackage{graphicx}
\begin{document}
\theoremstyle{plain}
  \newtheorem{theorem}{Theorem}
  \newtheorem{conjecture}{Conjecture}
  \newtheorem{proposition}{Proposition}
  \newtheorem{lemma}{Lemma}
  \newtheorem{corollary}{Corollary}

\theoremstyle{remark}
  \newtheorem{remark}{Remark}
  \newtheorem{remarks}{Remarks}

\theoremstyle{definition}
  \newtheorem{definition}{Definition}

\include{psfig}

\author{Sergiu Klainerman}
\address{Department of Mathematics, Princeton University,
 Princeton NJ 08544}
\email{ seri@math.princeton.edu}
\title[Trapped surfaces]{A   Fully Anisotropic Mechanism for Formation of Trapped Surfaces in Vacuum }

\author{Jonathan Luk}
\address{Department of Mathematics, University of Pennsylvania, Philadelphia PA 19104}
\email{jwluk@sas.upenn.edu}

\author{Igor Rodnianski}
\address{Department of Mathematics, MIT, 
Cambridge, MA 02139}
\email{ irod@math.mit.edu}
\vspace{-0.3in}
\begin{abstract}
We  present a new, fully anisotropic,   criterion for formation of trapped surfaces in vacuum.
More precisely we   provide  local   conditions   on  null data,  concentrated 
in a neighborhood of  a   short null geodesic segment (possibly   flat in all other directions)    whose future development
contains a trapped  surface. This     extends    considerably   the   previous result of 
 Christodoulou \cite{Chr:book}     which     required  instead    a    uniform  condition along all   null geodesic  generators.     To obtain our result we combine Christodoulou's mechanism for the formation of a trapped surface with a new   deformation process which takes place along  incoming null
  hypersurfaces.
\end{abstract}
\maketitle
\section{Introduction}
According to the  celebrated  incompleteness  result of   Penrose,     the   future   Cauchy  development of a     non-compact initial data
 set of the Einstein  vacuum\footnote{The result  of    Penrose applies in fact  
   to the more general   Einstein-matter  equations satisfying the null energy condition,  but  we restrict  our  considerations   here  to   the vacuum case.}
      equations, 
      $$\mbox{Ric}(g)=0$$
           which contains a \textit{trapped} surface,  must be    incomplete.    
   Thus, in a sense,   the fundamental issue       of  formation  of   spacetime   singularities   in gravitational collapse 
        is  reduced   to the somewhat    more tangible       problem   of formation 
    of    trapped surfaces.   This, on the other hand, is still  a  highly  non-trivial      problem.   Indeed, the expansions  of both null geodesic congruences generated by a  compact, trapped  surface  $S$  is required, by definition,  to be negative at every point on $S$.
    To  show that  such   surfaces can form in evolution,   starting   with   regular  initial data sets  which contain no
    trapped  surfaces,  requires    a deep   understanding    of the    dynamics of the    Einstein equations.  It is for this  reason
    that the  problem  has  remained open for     more than forty years,  in the wake of   Penrose's result,   until the  recent    breakthrough 
    of Christodoulou.   In      \cite{Chr:book}    he was able to  identify  an open  set of regular\footnote{Smooth and  free
    of   trapped surfaces.}  initial conditions,   on a finite outgoing null hypersurface, with  trivial data  on an incoming      null hypersurface,        whose  future development must  contain a trapped surface.   The main  condition in   Christodoulou's    result is that the      data  verify a     uniform   lower bound   condition,   with respect to all   short,    null  geodesic  generators of the  outgoing  initial      null hypersurface.
 The goal of this paper is to  significantly  relax    this uniform condition by showing    that  a trapped surface  forms 
   even  if   the   local    null outgoing  data is   only 
    concentrated  in a neighborhood of  a   short null geodesic segment (possible   flat in all other directions).   \\
   
    We recall that Christodoulou's   proof  in \cite{Chr:book}
 rests on two main ingredients:   
 \begin{enumerate}
 \item    A semi-global existence result  for the   characteristic   initial   value problem
   with large   initial  data\footnote{On the outgoing null hypersurface. The incoming data is flat. }   measured      relative  to a small parameter $\de>0$.  The precise  dependence on $\de$,  which Christodoulou  calls 
   the  short  pulse method,  was        subsequently   relaxed in  \cite{K-R:trapped}, \cite{K-R:scarred};
    see also \cite{R-T}.
    In all  these  results  the data on the incoming null hypersurface is assumed to be flat. This   restriction has been recently removed in  \cite{L-R}.  
    
     The semi-global result  allows one to construct  the   future development of the  initial data,  together with a double null foliation\footnote{Such that     the initial configuration is given by   the  incoming $\{\ub=0\}$ and outgoing   $\{u=0\} $    initial   null hypersurfaces.}  $(u,\ub)$,   $0\le u\le  u_*$,  $0\le\ub\le  \ub_*$,     and     full control  on  all the geometric quantities  associated to it.
   
   \item  An  amplification mechanism for the  integrals        $\int_\ga|\chih|^2 $   along  
    outgoing null geodesic segments $\ga$,     (with $\chih$   denoting the outgoing   null  shear). This mechanism,   which requires 
      the estimates obtained in the constructive step (1),   combines 
         with  a  uniform lower bound      assumption of these integrals  on the 
    initial   null hypersurface,  and leads to the formation
     of a trapped surface. It is important to note     that   all       trapped surfaces  found    in this fashion belong to the concentric spheres  generated by  
      the double null foliation $(u,\ub)$, i.e.   are  of the form $S=\{u=u_1, \ub=\ub_1\}$, for some
      $0 <u_1\le u_*$, $0<\ub_1\le\ub_*$. 
    
 \end{enumerate}
The new  result we present in this  paper    relies  heavily on the  \textit{hard} part of the 
above results,  i.e. the construction of the spacetime    in (1).   We modify however     part (2)
 by   combining   Christodoulou's argument with a  new  deformation argument along the incoming null
  hypersurfaces  $\{\ub=\mbox{const}\}$.  This allows us to dramatically weaken    his uniform condition
    merely to a localized condition in a neighborhood of a null geodesic   of $\{u=0\}$.   The deformation
     is determined  by       solving       a        simple      elliptic  inequality       on the     standard sphere $S_{0,0}=\{\ub =0\} \cap   \{u=0\}$,
      see \eqref{maineqR}.    We note that
    the trapped surface  we find by    our argument does  not  belong any longer to the double
     null foliation constructed  in step (1).

\subsection{Geometry of a  double null foliation} As in   \cite{K-R:trapped}
   we consider a region $\DD=\DD(u_*,\ub_*)$ of a vacuum spacetime $(M,g)$,
   \beaa
   \mbox{Ric}(g)=0,
   \eeaa
    spanned by a double null foliation  generated by the optical functions  $(u,\ub)$  increasing towards the future\footnote{These can be compared to the optical functions $u=\frac{t-r+1}{2}$, $\ub=\frac{t+r-1}{2}$ in Minkowski spacetime.},  where $0\le u\le u_*$ and $0\le\ub\le  \ub_*$ (see Figure 1). These spacetimes will be constructed via solving a characteristic initial value problem.

\begin{minipage}[!t]{0.4\textwidth}
    \includegraphics[width=4.2in]{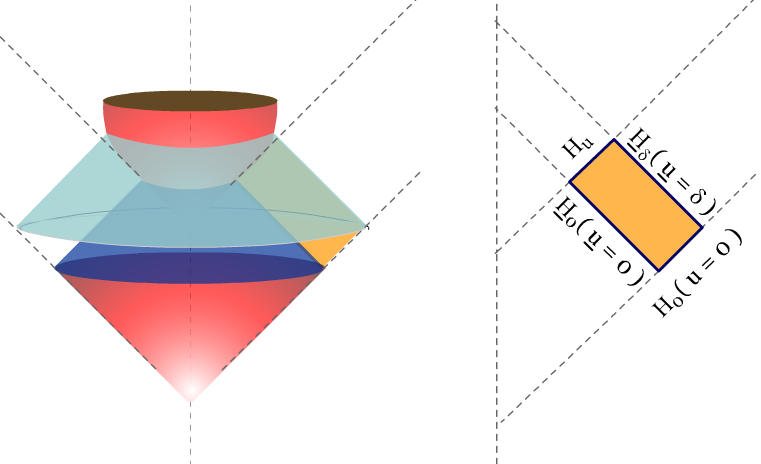}
\end{minipage}
\hspace{0.2\textwidth}
\begin{minipage}[!t]{0.4\textwidth}
The shaded region  on the right represents the domain $\DD(u_*,\ub_*)$,
$\ub_*= \de$.  The same picture is represented   on the left, by    emphasizing      that all points  in the
 diagram on the right are in fact  $2$-surfaces.
\end{minipage}
    We denote by $H_u$ the outgoing  null hypersurfaces generated by the  level surfaces of $u$  and     by  $\underline{H}_{\ub}$ the incoming  null hypersurfaces generated  level hypersurfaces  of  $\ub$.  We write 
$S_{u,\ub}=H_u\cap \underline{H}_{\ub}$ and   denote by $H_{u}^{(\ub_1,\ub_2)}$  and $\underline{H}_{\ub}^{(u_1,u_2)}$ the regions of these null hypersurfaces   defined by $\ub_1\le\ub\le\ub_2$ and respectively $u_1\le u\le u_2$.
    Let $L=-2g^{\a\b}\pr_\a u\pr_\b,\,\, \Lb=-2g^{\a\b}\pr_\a \ub\pr_\b,\,\,$ be the geodesic vectorfields associated to the two foliations and  define, 
    \bea
    g(L,\Lb):=-2\Om^{-2}=4g^{\a\b}\pr_\a u\pr_\b\ub. \label{eq:def.omega}
    \eea
As is well known,  the  space-time slab  $\DD(u_*, \ub_*)$  is completely 
determined  (for small values of $u_*, \ub_*$)   by  data along the null  hypersurfaces $H_0$,  $\Hb_0$ corresponding to
 $\ub=0$ and $u=0$ respectively.
We assume that
 $\Om=1$ along $H_0$ and $\underline{H}_0$, i.e.,
\bea
\Om(0,\ub)=1, \qquad 0\le \ub\le\ub_*,\\
\Om(u,0)=1, \qquad 0\le u\le u_*.
\eea
 We denote  by  $r=r(u,\ub)$  the area-radius of the $2$-surface $S=S_{u,\ub}$,  i.e. $|S_{u,\ub}|=4\pi r(u,\ub)^2$. In this paper, we assume $r(0,0)=1$. \footnote{General values for $r(0,0)$ can be recovered from this special case via a simple rescaling argument.}

   Throughout this paper we   work with the normalized null pair $(e_3,e_4)$ defined by
$$
e_3=\Omega\Lb,\quad e_4=\Omega L$$
which satisfy
$$\qquad g(e_3,e_4)=-2.$$
Given a   2-surfaces  $S_{u,\ub}$  and   $(e_a)_{a=1,2}$   an arbitrary  frame tangent to it  we define   the Ricci coefficients,
 \bea
\Ga_{(\la)(\mu)(\nu)}=g(e_{(\la)}, D_{e_{(\nu)}} e_{(\mu)} ),\quad \la,\mu ,\nu =1,2,3,4.
 \eea
 These coefficients are completely determined by the following components,
 \begin{equation}
\begin{split}
&\chi_{ab}=g(D_a e_4,e_b),\, \,\, \quad \chib_{ab}=g(D_a e_3,e_b),\\
&\eta_a=-\frac 12 g(D_3 e_a,e_4),\quad \etab_a=-\frac 12 g(D_4 e_a,e_3)\\
&\omega=-\frac 14 g(D_4 e_3,e_4),\quad\,\,\, \omegab=-\frac 14 g(D_3 e_4,e_3),\\
&\ze_a=\frac 1 2 g(D_a e_4,e_3)
\end{split}
\end{equation}
where $D_a=D_{e_{(a)}}$. We also introduce the  null curvature components,
 \begin{equation}
\begin{split}
\a_{ab}&=R(e_a, e_4, e_b, e_4),\quad \, \,\,   \aa_{ab}=R(e_a, e_3, e_b, e_3),\\
\b_a&= \frac 1 2 R(e_a,  e_4, e_3, e_4) ,\quad \bb_a =\frac 1 2 R(e_a,  e_3,  e_3, e_4),\\
\rho&=\frac 1 4 R(Le_4,e_3, e_4,  e_3),\quad \si=\frac 1 4  \,^*R(e_4,e_3, e_4,  e_3).
\end{split}
\end{equation}
Here $\, ^*R$ denotes the Hodge dual of $R$.  We denote by $\nab$ the 
induced covariant derivative operator on $S_{u,\ub}$ and by $\nab_3$, $\nab_4$
the projections to $S_{u,\ub}$ of the covariant derivatives $D_3$, $D_4$.
Observe that,
\begin{equation}\label{eqnOm}
\begin{split}
&\omega=-\frac 12 \nab_4 (\log\Omega),\qquad \omegab=-\frac 12 \nab_3 (\log\Omega),\\
&\eta_a=\zeta_a +\nab_a (\log\Omega),\quad \etab_a=-\zeta_a+\nab_a (\log\Omega).
\end{split}
\end{equation}
We recall the integral formulas\footnote{see for example  Lemma 3.1.3 in \cite{KNI:book}.}
for a  scalar function $f$ in $\DD$,
\bea
\frac{d}{ d\ub} \int_{S_{u,\ub}} f &=& \int_{S_{u,\ub}} \big(\frac{df}{d\ub}+ \Om \trch  f\big)= \int_{S_{u,\ub}} \Om \big(e_4(f)+ \trch  f\big)\nn,\\
\frac{d}{ du } \int_{S_{u,\ub}} f &=&  \int_{S_{u,\ub}} \big(\frac{df}{du}+\Om \trchb  f\big)= \int_{S_{u,\ub}}\Om \big(e_3(f)+ \trchb f\big). \label{form:transp-f-u}
\eea
In particular,
\bea
\frac{dr}{d\ub}=\frac{1}{8\pi}\int_{S_{u,\ub}}\Om\trch, \qquad \frac{dr}{ du }=\frac{1}{8\pi}\int_{S_{u,\ub}}\Om\trchb.\label{deriv-r}
\eea
We also recall the  following commutation formulae\footnote{Here, $\in_{ab}$ denotes the area 2-form and $\div$    the divergence operator  on the sphere $S_{u,\ub}$. We refer the readers to \cite{K-R:trapped} for the precise definitions of all of the notations below.} between $\nab$ and $\nab_4, \nab_3$ in \cite{KNI:book}:
\begin{lemma}
\label{le:comm}
For a scalar function $f$:
\bea
\,[\nab_4,\nab] f&=&\frac 12 (\eta+\etab) D_4 f -\chi\cdot \nab f,\label{comm:nabnab4}\\
\,[\nab_3,\nab] f&=&\frac 12 (\eta+\etab) D_3 f -\chib\cdot \nab f.
\label{comm:nabnab3}
\eea

For a 1-form  tangent to $S$:
\beaa
\,[\nab_4,\nab_a] U_b&=& -\chi_{ac} \nab_c U_b +\in_{ac}\, ^*\beta_b U_c   
+\frac 12(\eta_a+\etab_a) D_4 U_b \\
&-&\chi_{ac} \,\etab_b\,  U_c +\chi_{ab} \,\etab\cdot U,\\
\,[\nab_3,\nab_a] U_b&=&-\chib_{ac} \nab_c U_b    +\in_{ac} ^*\betab_b U_c +
\frac 12(\eta_a+\etab_a) D_3 U_b   \\
 &-& \chib_{ac}  \eta_b\,U_c+\chib_{ab} \, \eta\cdot U.
\eeaa
In particular,
\beaa
\,[\nab_4,\div] U&=&-\frac 1 2 \trch  \, \div U-\chih\c \nab U   -\b\c U +\frac 1 2 (\eta+\etab)\c \nab_4 U-   \etab\c \chih\c U ,\\
\,[\nab_3,\div] U&=&-\frac 1 2 \trchb  \, \div U-\chibh\c \nab U +  \bb\c U +\frac 1 2 (\eta+\etab)\c \nab_3 U
-   \eta\c \chibh\c U .
\eeaa
\end{lemma}

\subsection{Main theorem}
For simplicity of our presentation\footnote{Similar results    can be  derived using  the  classes of initial data   discussed in 
 \cite {K-R:trapped},  \cite {K-R:scarred}  and \cite{L-R}.}  we   describe our main  result in the context  of 
 the class of initial data  used by Christodoulou. This class of initial data gives rise to a class of spacetimes endowed with a double null foliation as described in the previous subsection.  As in  \cite{Chr:book},
   we prescribe  the  null  incoming    data      to    be   trivial, i.e. corresponding   to  null cones 
   in Minkowski space.
In particular  $S_{0,0}$ is  the standard sphere of radius $1$.  To prescribe    null    data   on     $H_0$  amounts    to prescribe an arbitrary symmetric, traceless, smooth tensor     $\chih_0$,
 called initial shear. In \cite{Chr:book}, the initial shear is prescribed only in a region with a short characteristic length scale, i.e., $0\leq \ub \leq \de$.

 \begin{definition}
    Given $\de>0$ and $B>0$, we say that   a  smooth   shear     $\chih_0$,   supported on $H_0^{(0,\de)}$,  verifies     Christodoulou's  $\de$-short pulse condition with constant $B$ if,
     \bea
     \label{Chr-de.cond}
 \sup_{\ub}\sum_{i\leq 5}\sum_{k\leq 3}\delta^{\frac 12+k}||\nab_4^k\nab^i\chih_0||_{L^\infty(S_{0,\ub})}\leq B.
\eea
\end{definition}
It is proved in \cite{Chr:book} that given any $B>0$ and $u_*<1$, there exists $\de$ sufficiently small such that any 
prescribed shear  $\chi_{0}$  satisfying             \eqref{Chr-de.cond}     gives rise to a unique smooth spacetime $\mathcal D(u_*,\de)$. In particular, the final inner sphere $S_{u_*,0}$ has radius 
$$r_* = 1-u_*,$$
which can be arbitrarily small, as long as $\de$ is also chosen to be sufficiently small. This is summarized in the following theorem:

\begin{theorem}[Christodoulou\footnote{Strictly speaking, the original theorem of Christodoulou is stated somewhat differently. The version we use here  is a straightforward reformulation of 
  his  results    in    \cite{Chr:book}  by a simple rescaling.} \cite{Chr:book}]\label{Chr.existence.thm}
For every $B>0$ and $u_*<1$ there exists $\de>0$ sufficiently small such that if the  null shear $\chih_0$   along $H_0$    verifies    Christodoulou's     $\de$-short pulse condition with constant $B$, then the   future  development  of the    corresponding   initial data  contains    $\DD(u_*,\de)$     as a regular  spacetime    in   which,   in particular, the   conditions {\bf MA1}--{\bf MA4} below   are satisfied,      with $\de_0=\de^{1/2}$.
\end{theorem} 

  \begin{enumerate}
 \item[\bf MA1.] $\Om $ is  comparable with its  initial value
\beaa
\Om= 1 +O(\de_0).\\
\eeaa
   
 \item[\bf MA2.]  The  Ricci coefficients $\chi, \om, \eta, \etab, \nab(\log\Om),  \chib, \omb$   verify

 \beaa
 |\chih, \om|&=&O(\de^{-1/2}), \\ 
    |\trch  |&=&O(1), \\ 
      |\eta, \etab, \chibh, \trchb+\frac{2}{r}, \omb, \nab(\log\Om)|&=& O(\delta_0).\\
 \eeaa

 \item[\bf MA3.]  The derivatives of Ricci coefficients satisfy
 \beaa
             |\nab\eta|&=&O(\delta_0 \de^{-1/2}), \\
             |\nab\chibh,  \nab\trchb,  \nab\omegab|&=& O(\delta_0).
   \eeaa
 \item[\bf MA4.] $\trch$ is close to its Minkowskian value on the initial cone $\Hb_0$
 
 \beaa
     |\trch-\frac 2 r|&=&O(\de_0),\quad\mbox{on }\Hb_0. \\ 
  \eeaa

\end{enumerate}

To show that a trapped surface forms  in $\DD(u_*,\de)$, Christodoulou needs in addition a uniform lower bound   on the function  $M_0=M_0[\chih_{0}]$  defined on $S_{0,0}$   as follows,
\bea
M_0(\om)=M_0[\chih_0](\om):= \int_0^{\de} |\chih_0|^2(\ub',\omega) d\ub',    \label{eq:main-thm1}
\eea
where the integral is taken along the null geodesic generators on $H_0$, transversal  to $S_{0,0}$, initiating at 
 $\om\in S_{0,0}$.   More precisely,  he proved
\begin{theorem}[Christodoulou \cite{Chr:book}]\label{Chr.form.trapped.surface.thm}
Assume that  the  initial   null shear $\chih_0$   along $H_0$    verifies both     \newline\noindent  Christodoulou's     $\de$-short pulse condition with constant $B$ and   the    isotropic  condition 
\bea
 \inf_{\om\in  S_{0,0}  }    M_0(\om)\geq M_* >0.\label{main-cond.C}
 \eea
 Then,   for any given $r_*>0$, 
$$r_*<\frac{M_*}{2},\quad u_*=1-r_*,$$
there exists $\de>0$ sufficiently small (depending only on $B$ and $M_*$) such that     surface   $S_{u_*,\de}\subset \DD(u_*,\de)$, with 
$\DD(u_*,\de)$     constructed     by Theorem  \ref{Chr.existence.thm},  is necessarily trapped.

\end{theorem} 
We now state our main result, which replaces the condition of the uniform lower bound \eqref{main-cond.C} by merely the condition that $M_0$ is positive somewhere:
\begin{theorem}[Main theorem]
\label{Main:thm}
Assume that  the initial   null shear $\chih_0$   along $H_0$    verifies both   \newline\noindent    Christodoulou's     $\de$-short pulse condition with constant $B$ and the following anisotropic condition
\bea
\sup_{\om\in  S_{0,0}  } M_0(\om)>0.\label{main-cond}
\eea
 Then,  there exists $u_*>0$ sufficiently close to $1$ and $\de>0$ sufficiently small (depending on $B$ and the function $M_0$) such that the future development   $\DD(u_*,\de)$,
   constructed    by  Theorem  \ref{Chr.existence.thm},
 contains a trapped surface.

\end{theorem}

By continuity of $M_0$, the condition \eqref{main-cond} implies that there exists $\ep$ and $M_*$ such that
  \bea
\inf_{\om\in B_p(\ep)  }    M_0(\om)\geq M_* >0,  \label{eq:main-thm2}
\eea
  where   $B_p(\ep)$  is     a geodesic ball  of radius  $\ep$  around some $p\in S_{0,0}$.
 Our main theorem therefore follows from the more quantitative version below:

\begin{theorem}\label{main.thm.quan}
Assume that  the initial   null shear $\chih_0$   along $H_0$    verifies both    Christodoulou's     $\de$-short pulse condition with constant $B$ and   the    non-isotropic  condition     \eqref{eq:main-thm2}  with constants $\ep$ and
$M_*$.  Then,  for  any given $r_*>0$ verifying
\bea
r_*= c_0M_*\ep^5,\quad u_*=1-r_*,
\eea
(where $c_0$ is a universal constant) there exists $\de>0$ sufficiently small (depending on $B$, $M_*$ and $\ep$) such that the future development   $\DD(u_*,\de)$,
   constructed    by  Theorem  \ref{Chr.existence.thm},
 contains a trapped surface of area at least $\gtrsim M_*^2 \ep^{10}$.    

\end{theorem}

\begin{remark}
In Chapter 2 of \cite{Chr:book}, Christodoulou also constructed a class of initial data satisfying the assumptions of Theorems \ref{Chr.existence.thm} and \ref{Chr.form.trapped.surface.thm}. The construction, which is based on solving ordinary differential equations along the null generators of the initial outgoing null hypersurface, can also be applied to construct initial data which obey the assumptions of our main theorem (Theorem \ref{Main:thm}). In particular, since the initial data at any point $p$ on the initial hypersurface depend only on the prescribed conformal part of the metric along the null generator passing through $p$, we can prescribe initial data which are flat except in a small neighborhood of an outgoing generator on the initial hypersurface.
\end{remark}

\begin{remark}    The  conditions  {\bf MA1}--{\bf MA4}   are     needed
 to implement  part  (2)  of   our main theorem.   Though strictly speaking  Christodoulou's
  theorem implies  {\bf MA1}--{\bf MA4}   with $\de_0=\de^{1/2}$,  we  prefer to    write   them
  in this more general form   with respect  to   a second parameter $\de_0$,    such that  $\de_0$, $\de_0^{-1}\de$
   are sufficiently small.
  This formulation  allows   us   to  adapt  our result to the more general  initial data used  in 
  \cite{K-R:trapped}, \cite{K-R:scarred} and \cite{L-R}.  
  
\end{remark}

\begin{remark}
The original theorem of Christodoulou \cite{Chr:book} also applies    to the case when the initial   area-radius is $r_0$.
In that case,  \eqref{Chr-de.cond} can be replaced by 
$
\sup_{\ub}\sum_{i\leq 5}\sum_{k\leq 3}\delta^{\frac 12+k}r_0^{1+i}||\nab_4^k\nab^i\chih_0||_{L^\infty(S_{0,\ub})}\leq B,
$   
and    the definition of $M$      is similarly modified, i.e.  $M(\omega;r_0)=\int_0^{\delta} r_0^2|\chih_0|^2(\ub',\omega)d\ub'.$
With this  definition, if the initial data satisfy
$\inf_{\om\in  S_{0,0}  }    M_0(\om;r_0)\geq M_* >0,$
the formation of trapped surface theorem (Theorem \ref{Chr.form.trapped.surface.thm}) still holds after choosing $\delta$ sufficiently small depending on $B$ and $M_*$.
By the same rescaling argument   our main theorem (Theorem \ref{Main:thm}) holds in the case where the initial area-radius is $r_0$ if the condition
$\sup_{\om\in  S_{0,0}  }    M_0(\om;r_0) >0$ is verified.
Moreover, in \cite{Chr:book}, Christodoulou proved that one can take $r_0\to \infty$ and construct a spacetime from null infinity. For such spacetimes, our anisotropic condition can be replaced by the condition 
$\sup_{\om\in  S_{0,0}  }    M_0(\om;\infty) >0$
and still guarantee the formation of a trapped surface.  
\end{remark}

\begin{remark}    Note that   while   in  \cite{Chr:book}  the desired          trapped  surface 
  can be found  among the surfaces 
  $S_{u,\ub}$,  consistent with the double null foliation,    this is no longer the case 
  in our theorem. Instead, we will identify a different 2-sphere, embedded in the null hypersurface $\{\ub=\delta\}$, which can be shown to be trapped.      Note also  that     the function   $M_0[\chih](\om)=\int_0^\de |\chih(\ub, \om)|^2 d\ub' $  is invariant\footnote{In other words 
$M_0[\chih](\om)=       M_0[\chih'](\om) =\int_{f(0,\om)}^{f(\delta,\om)} |\chih'(\ub',\om)|^2 d\ub'$    where  $\chih'_{ab}=\frac 12 g(D_{e_a} \frac{\partial}{\partial \ub'}, e_b)$       for a new foliation  $\ub'=f(\ub,\om)$,  with $\frac{\partial f}{\partial \ub}> 0$.    } with respect to  a change of foliation     along  the initial hypersurface $H_0$.  It is thus  impossible to change the foliation on the initial hypersurface $H_0$ so that the isotropic condition in Theorem \ref{Chr.form.trapped.surface.thm} holds.
  \end{remark}

We  show that a trapped surface exists in $\DD(u_*,\de)$ by finding an embedded trapped 2-sphere in the incoming null hypersurface $\Hb_\de$. Notice that by {\bf MA2}, $\trchb<0$ on $\Hb_{\de}$. Therefore, it suffices to find a 2-sphere such that the outgoing null expansion is also pointwise negative. We will achieve this in two steps. In Section \ref{sec.reduction} we make use of 
  the conditions {\bf MA1}--{\bf MA4} to 
  reduce the 
 problem of  existence of a trapped  surface to 
 that of finding    appropriate  solutions to an elliptic  inequality, see  \eqref{maineqR},    on 
 $S_{0,0}$.
  In Section 
\ref{sec.solution}, we  prove that under the assumption \eqref{eq:main-thm2} of the main theorem, a desired solution to the elliptic inequality exists.

\section{Reduction to an elliptic inequality on  the initial  sphere $S_{0,0}$}\label{sec.reduction}
In   this section, we  show that under the assumptions {\bf MA1}--{\bf MA4}, the existence of a trapped surface can be reduced to constructing a solution to an elliptic inequality \eqref{maineqR}.  The main result of the section is stated
 in Theorem \ref{thm5}.

We first make a remark  concerning    the global parametrization of points 
 in $ \DD(u_*,\delta)$ by $u$, $\ub$ and  coordinates $\om$  on $S_{0,0}$.
 
Proceeding as in \cite{Chr:book}  we associate  to 
  each coordinate patch on $S_{0,0}$, a system of transported coordinates
defined   by
\bea
\Ls_{\Omega e_4} \theta^a=0, \quad \mbox{on }H_0, \label{coord1}
\eea
and
\bea
\Ls_{\Omega e_3} \theta^a =0,\quad \mbox{in } \DD(u_*,\delta), \label{coord2}
\eea
where $\Ls$ is the restriction of the Lie derivative to $T S_{u,\ub}$ (see \cite{Chr:book}, chapter 1).
This provides an identification of each point in the spacetime $\mathcal D(u_*, \de)$
with a point in the initial sphere $S_{0,0}$ by the value of the coordinate functions.

It follows that any point in $\mathcal D(u_*, \de)$ can also  be uniquely specified by   the coordinates $(\ub, u, \omega)$,
where $\omega\in S_{0,0}$.

We can now state the main result of this section:

\begin{theorem}
\label{thm5}
Assume that the spacetime $\DD(u_*,\delta)$ satisfies {\bf MA1}--{\bf MA4}. Let $M_0$ be the function on the initial sphere $S_{0,0}$ defined by \eqref{eq:main-thm1}, i.e.,
$$M_0(\omega) := \int_0^{\delta} |\chih|^2(u=0,\ub',\omega) d\ub'.$$
Assume   $R$ is a smooth function on $S_{0,0}$  satisfying $r_*+C\de_0 < R< 1$ as well as 
 the  elliptic inequality\footnote{Here $\lap_0$ and $\nab_0$ are defined with respect to the connection on the initial sphere $S_{0,0}$.} on $S_{0,0}$
\bea\label{maineqR}
-\lap_0 R+ R^{-1}|\nab_0 R|^2 + R <2^{-1}M_0-C\de_0,
\eea
with  $C>0$  a constant  depending  only on $\displaystyle\sum_{i\leq 2} ||\nab_0^i R||_{L^\infty(S_{0,0})}$.\\
Then, for $\delta\delta_0^{-1}$, $\delta_0$ sufficiently small,
the $2$-sphere defined by $\{(\ub,u,\omega) : \ub=\delta, 1-u = R(\om)\}$ is a trapped surface.
\end{theorem}

We  prove Theorem \ref{thm5} in this section and leave for the  next section 
the task to show that a smooth solution $R$ to \eqref{maineqR} exists,  
provided that $M_0$ satisfies the assumptions of  our  main theorem.

The proof of Theorem \ref{thm5} will be achieved in two steps. In Section \ref{Chr-heuristic}, we will carry out Christodoulou's argument in \cite{Chr:book} to estimate $\trch$ with respect to the original foliation on $\Hb_{\de}$. In Section \ref{secmaintrans}, we will then deform the foliation on $\Hb_{\de}$ in such a way that the 2-sphere $\{(\ub,u,\omega) : \ub=\delta, 1-u = R(\om)\}$ is a level surface adapted to the new foliation. We then compute the desired outgoing expansion $\trch'$ on the 2-sphere $\{(\ub,u,\omega) : \ub=\delta, 1-u = R(\om)\}$.

We note that the derivation in Section \ref{secmaintrans} can be  simplified   by taking into account    the     specific properties  of      the spacetime
 constructed  in Theorem \ref{Chr.existence.thm}. We will include the simplified argument in Section \ref{sec.simp.der}. Nevertheless,  
  we prefer to proceed  below with the  general derivation in view of possible   applications
    to a more general setting.

\subsection{Christodoulou's argument}\label{Chr-heuristic}

In \cite{Chr:book}, it was shown that under the assumptions {\bf MA1}--{\bf MA4},  the expansion $\trch$ on each of the spheres  $S_{u, \de}$  on $\underline{H}_{\delta}$ can be computed up to a small error depending on $\delta$. In the context of \cite{Chr:book}, where a uniform lower bound on $M_0$ is assumed, this is sufficient to conclude the existence of a trapped surface $S$ of 
the form $S=S_{u,\de}$.  In the case  of our weaker condition  \eqref{eq:main-thm2}, his argument only shows that $\trch$ becomes sufficiently negative in part of the sphere $S_{u,\delta}$.  To obtain a trapped surface we need    to combine that fact     with a  new deformation argument   of   the   foliation on $\Hb_{\delta}$.  

Christodoulou's   argument for the formation of trapped surfaces in  \cite{Chr:book}  rests on the equations\footnote{ The operator $\widehat{\otimes}$   referes to    the traceless part of the symmetrized  tensor product.        We refer the readers to \cite{K-R:trapped} for the precise  notations used in this paper.},
 \beaa
\nab_4 \trch&=&-|\chih|^2   -\frac 12 (\trch)^2 -2\om \trch \\
\nab_3\chih+\frac 1 2 \trchb \chih&=&\nab\widehat{\otimes} \eta+2\omb \chih-\frac 12 \trch \chibh +\eta\widehat{\otimes} \eta
\eeaa
In view of our Ricci coefficients  assumptions we   can rewrite,
\beaa
\nab_4 \trch&=&-|\chih|^2 +O(\de^{-1/2}) \\
\nab_3\chih+\frac 1 2 \trchb \chih&=&    O(\de_0 \de^{-1/2})            
\eeaa
Multiplying  the second equation by $\chih$, 
\beaa
\nab_3  |\chih|^2+\trchb|\chih|^2&=&O(\de_0 \de^{-1})
\eeaa
 Using also our assumptions for  $u, \ub, \Om$    we deduce\footnote{Here, $\frac{\partial}{\partial\ub}$ and $\frac{\partial}{\partial u}$ refers to the coordinate vector fields defined with respect to a coordinate system $(u,\ub,\th^1,\th^2)$, where $\th^1$ and $\th^2$ satisfy the conditions \eqref{coord1} and \eqref{coord2}. Notice that this is different from the vector fields $\Omega e_4$ and $\Omega e_3$ as $\Omega e_4$ would have an "angular component" in the coordinate system $(u,\ub,\th^1,\th^2)$. In particular, $\frac{\partial}{\partial\ub}$ is not parallel to the null generators. Nevertheless, as a consequence of the assumptions {\bf MA1}--{\bf MA4}, the difference between $\frac{\partial}{\partial\ub}$ and $\Omega e_4$ is a vector field on $TS_{u,\ub}$ with a small norm and will be collected in the error terms in the equations.},
\bea
\frac{\partial}{\partial\ub}\trch&= &-|\chih|^2+O(\de^{-1/2})  \label{eq:intr.nab4trchi}\\
\frac{\partial}{\partial u}|\chih|^2+\trchb|\chih|^2&=&O(\delta_0\de^{-1})    \label{eq:intr.nab3chi}
\eea
Integrating \eqref{eq:intr.nab4trchi} we    obtain,
\bea
\trch(u,\ub)&= &\frac{2}{r(u,0)}-\int_0^{\ub}   |\chih|(u,\ub')^2 d\ub'   +O(\de_0) \label{eq.ineq.intrch}
\eea

In view of our assumptions for $\trchb$
and $\frac{dr}{du} $, \eqref{eq:intr.nab3chi} implies
\beaa
\frac{d}{du} (r^2 |\chih|^2)&=&r^2\frac {d}{du}|\chih|^2+2 r \frac {dr}{du} |\chih|^2=r^2\big[-\trchb |\chih|^2 +O(\de_0\de^{-1})\big]   +2r\big[  -1+O(r\de_0)  \big] |\chih|^2  \\
&=&O(\de_0\de^{-1}).
\eeaa
Therefore,
\beaa
r^2 |\chih|^2(u,\ub)=r^2(0,\ub)  |\chih|^2(0 ,\ub)+O(\de_0\de^{-1})
\eeaa
Let $\chih_0$ denote the initial data for $\chih$:
\bea
\chih_0(\ub)=\chih(0,\ub).
\eea

We  deduce, 
\beaa
|\chih|^2(u,\ub)&=&\frac{ r^2(0,\ub)}{r^2(u,\ub)}   |\chih_0|^2(\ub)+O(\de_0\de^{-1}  )
\eeaa
Since $r(u,\ub)=r(u,0)+O(\de)$,
\beaa
|\chih|^2(u,\ub)=\frac{ 1}{r^2(u,0)}   |\chih_0|^2(\ub)+O(\de_0\de^{-1}).
 \eeaa
Thus, returning to  \eqref{eq.ineq.intrch}, and recalling that
$$M_0(\omega)=\int_0^{\de}  |\chih_0|^2(\ub',\omega) d\ub',$$
we deduce the following:  
\begin{proposition}\label{prop:Christ}
Under the assumptions  {\bf MA1}--{\bf MA4} we have, for $\de_0$, $\de_0^{-1}\de$ sufficiently small,
\bea
\label{form-prop:christ}
\trch(u,\ub=\de,\om)=\frac{2}{r(u, 0)} -
\frac{ 1}{r^2(u,0)}M_0(\om)+ O(\de_0)
\eea

\end{proposition} 
Since $r(u,0)=1-u+O(\de_0)$, this implies
\begin{corollary}
For $\de_0$, $\de_0^{-1}\de$ small, the necessary and sufficient condition to have $\trch< 0$
everywhere on the sphere $S_{u,\delta}$ is that
\bea
{2(1-u)}<  M_0(\om)- O(\de_0) \label{Cond.Chr}
\eea
holds uniformly for every $\om\in S_{0,0}$.
\end{corollary}

Under the assumptions of our main theorem, we can only hope that the outgoing null expension 
$\trch$ adapted to the foliation $(u,\ub)$ becomes negative 
in the part where $M_0$ is positive. Thus to prove our main theorem, we need to combine this argument with the new deformation mechanism 
which leads to the formation of a trapped surface that is no longer adapted to the double null foliation
$(u,\ub)$. Instead, as stated in Theorem \ref{thm5}, the trapped surface will be a topological 2-sphere 
embedded in the incoming null hypersurface $\{\ub = \de\}$ defined by $\{(\ub,u,\omega): \ub=\de, 1-u=R(\omega)\}$.

 \subsection{Main transformation formula}\label{secmaintrans}      
 According to the statement of Theorem \ref{thm5}, $\{(\ub,u,\omega): \ub=\de, 1-u=R(\omega)\}$ will correspond to 
 a trapped surface provided that $R$ satisfies \eqref{maineqR}. To verify that, we need to compute
 its null expansion, which differs from the null expansion $\trch$ relative to
 the double null foliation $(u,\ub)$ restricted to $\{(\ub,u,\omega): \ub=\de, 1-u=R(\omega)\}$.
 To compute the correct null expansion $\trch'$,  we introduce  the new null  frame adapted to this set,
 \bea
 \label{newframe}
e_3'=e_3,\qquad   e_a'=e_a- \Om   e_a(R)e_3, \qquad e_4'=e_4-2\Om e^a(R) e_a+    \Om^{2} |\nab R|^2 e_3
 \eea
 Recall that by definition $e_3(u)=\Om^{-1}$. Thus we have, $e_a'(u+R-1)=e_a(R)- e_a(R) \Om e_3(u)=   0$.   Also,
 since $e_3$ is orthogonal to any vector tangent to $\Hb$ we 
easily check that
\beaa
g(e_a', e_b')=g(e_a,e_b)=\de_{ab},\quad g(e_4', e_a')=
g(e_4', e_4')=0,\quad 
g(e_3', e_4')=-2.
\eeaa
We prove the following

\begin{proposition} The trace of the null second fundamental form $\chi'$,    relative to the  new frame  \eqref{newframe},  is given by 
 \bea
\trch'&=&\trch -2\Om \Delta R -4\Om\eta\cdot\nab R-4\Om^2\chibh_{b c} \nab^b R\nab^c R -\Om^2\trchb|\nab R|^2-8\Om^2\omb |\nab R|^2.\label{pr-trch}
 \eea
\end{proposition}

\begin{proof}
Let $F_a=\Om(\nab_a R)$. We can then write $e_4'=e_4-2F+|F|^2 e_3$ with  $F=F^c e_c$
 and  $e_b'=e_b-F_b e_3$.
We have,
\beaa
\chi'(e_a', e_b')&:=&g(D_{a'} e_4', e_b')=g(D_a e_4', e_b')- F_a g(D_3 e_4', e_b')
\eeaa
The first term is given by
\beaa
g(D_a e_4', e_b')&=&g\big(D_a(e_4-2F +|F|^2  e_3)\,,  \,e_b- F_b e_3\big)\\
&=&\chi(e_a, e_b)-2F_b \ze_a-2\nab_a F_b +2 F_b g(D_a F, e_3) +|F|^2g(D_a e_3,  e_b-F_b e_3)\\
&=&\chi_{ab}-2\ze_a F_b-2\nab_ a F_b-2 F_b\,  \chib(F, e_a)+|F|^2\chib_{ab}\\
&=&\chi_{ab}-2\ze_a F_b-2\nab_ a F_b-2 F_b\,  F^c \, \chib_{ac}+|F|^2\chib_{ab}
\eeaa
Also,
\beaa
g(D_3 e_4', e_b')&=& g\big(D_3(e_4-2F+|F|^2  e_3)\,,\,  e_b- F_b e_3\big)\\
&=&g(D_3 e_4, e_b)-F_b g(D_3e_4, e_3)-2\nab_3 F_b\\
&=&2\eta_b+4F_b \omb -2\nab_3 F_b
\eeaa
Hence,
\beaa
\chi'_{ab}&=&\chi_{ab}-2\ze_b F_a-2\nab_ a F_b-2 F_b F^c \chib_{ac}+|F|^2\chib_{ab}- F_a \big(2\eta_b+4F_b \omb -2\nab_3 F_b)\\
&=&\chi_{ab}-2\nab_ a F_b+2 F_a \nab_3 F_b-2\ze_b F_a -2 F_a\eta_b+ |F|^2 \chib_{ab}-2 F_b F^c \chib_{ac} -4 \omb F_a F_b
\eeaa
By symmetry in $a, b$ we deduce the formula,
\bea
\chi'_{ab}=\chi_{ab}&- &(\nab_ a F_b+\nab_b F_a) + \nab_3(F_a F_b)  -(\ze_b  +\eta_b)F_a-(\ze_a+\eta_a )F_b\\
&+&  |F|^2 \chib_{ab}-   F_b F^c \chib_{ac}-    F_a F^c \chib_{bc} -4 \omb F_a F_b\nn
\eea
 and, taking the trace,
\beaa
\trch'&=&\trch -2\div F+\nab_3|F|^2-2(\eta+\ze)\c F+(|F|^2 \trchb -2\chib_{bc} F^b F^c)-4\omb |F|^2\\
&=&\trch -2\div F+\nab_3|F|^2-2(\eta+\ze)\c F -2\chibh_{bc} F^b F^c-4\omb |F|^2.
\eeaa
We next calculate  $\nab_3|F|^2$ using  
$$\nab_3 R=0$$
and the commutation formula
 \beaa
 [\nab_3, \nab] h =(\nab\log \Om) \nab_3 h-\chib\c\nab h.
\eeaa

Since  $F=\Om \nab R$  we deduce,
\beaa
\nab_3 F_a&=& \nab_3( \Om  \nab R)=
\Om  (\nab_3 \nab R)+(\nab_3\Om) \nab R\\
&=&\Om  \nab  \nab_3  R+(\nab\log\Om)\Om \nab_3 R- \Om \chib\c\nab R- 2\omb \Om \nab R\\
&=&- \Om \chib\c\nab R- 2\omb \Om \nab R.
\eeaa
from which we derive, 
\beaa
\nab_3|F|^2&=&-\Om^2 \trchb |\nab R|^2 -2\Om^2\chibh_{b c} \nab_b R \nab_c R  -4\Om^2\omb |\nab R|^2.
\eeaa
Therefore,
\beaa
\trch'&=&\trch -2\div F-2(\eta+\ze)\c F -2\chibh_{bc} F^b F^c-4\omb |F|^2\\
&&-\Om^2 \trchb |\nab R|^2 -2\Om^2\chibh_{b c} \nab_b R \nab_c R  -4\Om^2\omb |\nab R|^2\\
&=&\trch -2\Om \Delta R -4\Om\eta\cdot\nab R-4\Om^2\chibh_{b c} \nab_b R\nab_c R -\Om^2\trchb|\nab R|^2-8\Om^2\omb |\nab R|^2
\eeaa
as desired.

\end{proof}
We combine this with the main assumptions {\bf MA1}--{\bf MA4} to derive the following proposition:
\begin{proposition}\label{prop-summary-G}
The trace of the null second fundamental form $\chi'$,    relative to the  new frame  \eqref{newframe}, evaluated at the set $\{(\ub,u,\om):\ub=\de, 1-u=R(\om)\}$  is given by 
\bea
\trch'(u=1-R(\om),\ub=\de,\om)&=&\trch(u=1-R(\om),\ub=\de,\om) -\frac{2\Delta_0 R}{R^2} +\frac{2|\nab_0 R|^2}{R^3}+O_R(\de_0),\label{deform-trch}
 \eea
where $O_R(\de_0)$ denotes a term bounded by $C\de_0$, where $C$ depends only on the $L^\infty$ norm of $R$, $\nab_0 R$ and $\nab_0^2 R$, and $\nab_0$, $\Delta_0$, as before, are defined with respect to the connection on the initial sphere $S_{0,0}$.
\end{proposition}
\begin{proof}
By {\bf MA1}--{\bf MA4}, $\Om-1$, $\eta$, $\chibh$, $\trchb+\frac 2r$ and $\omb$ are small in terms of $\de_0$. Thus
$$\trch'(u=1-R(\om),\ub=\de,\om) = \trch(u=1-R(\om),\ub=\de,\om) -2 {\Delta R} +\frac{2|\nab R|^2}{R}+O_R(\de_0).$$
To achieve the proposition, we need to compare\footnote{Recall that $\nab R$ is defined with respect to the connection coefficients of the spheres adapted to the $(u,\ub)$ foliation at the set $\{(\ub,u,\omega): \ub=\de, 1-u=R(\omega)\}$; while $\nab_0 R$ is defined with respect to the connection coefficients on the initial sphere $S_{0,0}$.} $\nab R$ with $\nab_0 R$. To this end, we consider the equation 
$$\nab_3 R=0 $$
and commute with angular derivatives. The commutation formulae in Lemma \ref{le:comm} and {\bf MA1}--{\bf MA4} imply that
$$\nab_3 \nab R+\frac 12 \trchb \nab R = O_R(\de_0),$$
and 
$$\nab_3 \nab^2 R+ \trchb \nab^2 R = O_R(\de_0).$$
This implies, via the condition $|\trchb+\frac 2r|= O(\de_0)$ in {\bf MA2}, that
$$\nab_3 (r\nab R) = O_R(\de_0),$$
and 
$$\nab_3 (r^2\nab^2 R) = O_R(\de_0).$$
Therefore,
$$|\nab R|^2=\frac{|\nab_0 R|^2}{R^2}+O_R(\de_0),\quad \Delta R=\frac{\Delta_0 R}{R^2}+O_R(\de_0).$$
The conclusion thus follows.
\end{proof}

 We now combine the results of Propositions \ref{prop:Christ}
 and \ref{prop-summary-G}.  According to 
 Proposition \ref{prop:Christ}, $\trch$ evaluated at $\{(\ub,u,\omega): \ub=\de, 1-u=R(\omega)\}$ is given by
 \beaa
\trch(u=1-R(\om),\ub=\de,\om)=\frac{2}{R(u,0)} -
\frac{ 1}{R^2(u,0)}\int_0^{\de}  |\chih_0|^2(\ub',\omega) d\ub'+ O(\de_0).
\eeaa
Thus,  inserting in   \eqref{deform-trch}, we have
\beaa
\trch'(u=1-R(\om),\ub=\de,\om)&\le &\frac{2}{R}   - \frac{2\Delta_0 R}{R^2} +\frac{2|\nab_0 R|^2}{R^3}  -
\frac{ 1}{R^2}  M_0        -O_R(\de_0) .
\eeaa
where 
\beaa
M_0=   \int_0^{\de}  |\chih_0|^2(\ub') d\ub'.
\eeaa
 
This concludes the proof of Theorem \ref{thm5}.

\subsection{An alternative  derivation }   \label{sec.simp.der}
In the context of  our specific  situation  we outline\footnote{We thank Demetrios Christodoulou for pointing out a related idea.} here a slightly easier argument leading to the formula for $\trch'$ on $\Hb_{\de}$ in terms of $\trch$ on $\Hb_{\de}$ and the function $R$ in the limit $\de\to 0$. Note however that the derivation below depends on further properties of the Ricci coefficients, in addition to {\bf MA1}--{\bf MA4}, proved in \cite{Chr:book}. This provides an alternative proof of Theorem \ref{thm5}.

We recall from \cite{Chr:book} that in the limit $\de\to 0$, the following curvature components tend to zero:
\bea
\betab\to 0,\quad\alphab\to 0.\label{curv.zero}
\eea
Also,
\bea
\chibh=O(\de^{\frac 12}),\quad \trchb\to -\frac{2}{r},\label{chibh-limit}
\eea
$\chih$ is of the order $O(\de^{-\frac 12})$ and all the remaining Ricci coefficients are at least $O(1)$.

By \eqref{chibh-limit}, in the limit $\de\to 0$
$$\Ls_{\Om e_3} \gamma=2\chib=-\frac{2}{r} \gamma $$
holds along the incoming null hypersurface $\Hb_{\de}$, where $\gamma$ is a induce Riemannian metric on $S_{u,\de}$.

It follows that at the $\de\to 0$ limit, the metric $\gamma'$ on the sphere defined by $\{(\ub,u,\omega): \ub=\de, 1-u=R(\omega)\}$ is conformal to the metric on the standard sphere $\gamma_0$ with the conformal factor given by
\bea
\gamma' = R^2 \gamma_0.\label{conf.eqn}
\eea

Let $K'$  denote  the Gauss curvature of the sphere $\{(\ub,u,\omega): \ub=\de, 1-u=R(\omega)\}$ and $K_0$  be  the Gauss curvature on the initial sphere $S_{0,0}$. Then, in view of   \eqref{conf.eqn},   
\bea
K'=R^{-2}(K_0-\De_{0} (\log R))   = R^{-2}(1-\De_{0} (\log R))         .\label{eq:Gauss-conf}
\eea
Consider the Gauss equation for the sphere $\{(\ub,u,\omega): \ub=\de, 1-u=R(\omega)\}$,
\bea
K'=-\rho'+\frac 12 \chih'\cdot\chibh'-\frac 14 \trch'\trchb'.\label{Gauss-eq1}
\eea
Also the Gauss equation for the  $2$-spheres adapted to the original foliation,
\bea
K=-\rho+\frac 12 \chih\cdot\chibh-\frac 14 \trch\trchb.\label{Gauss-eq2}
\eea

We will derive $\trch'$ in terms of $\trch$ and $R$ using equations \eqref{Gauss-eq1} and \eqref{Gauss-eq2}.
By definition, $e_a'=e_a- \Om   e_a(R)e_3$ and $e_3'=e_3$. Since $e_3$ is normal to $e_a$, $e_a'$ and $e_3$, and
$$D_{e_3} e_3=-2\om e_3,$$
we have
$$\chib'_{ab}=\frac 12 g(D_{a'} e_3, e_b')=\frac 12 g(D_{a} e_3, e_b)=\chib_{ab}.$$
Therefore, as $\de\to 0$, 
$$\trchb'=\trchb\to -\frac 2r$$
and
\bea
|\chibh'|= O(\de^{\frac 12}).\label{chibhp.est}
\eea
On the other hand, $e_4'=e_4-2 \Om e_a(R) e_a+ \Om^{2} |\nab R|^2 e_3$, $\chih'$ can be expressed as $\chih$ plus a linear combination of the Ricci coefficients with respect to the original frame which are bounded independent of $\de$. Therefore,
\bea
|\chih'-\chih| = O(1),\label{chihp.est}
\eea
where the implicit constant may depend on $R$. Note that \eqref{chibhp.est} and \eqref{chihp.est} together imply in the limit $\de\to 0$,
$$|\chih\cdot\chibh-\chih'\cdot\chibh'|\to 0.$$

Since $e_4'=e_4-2 \Om e_a(R) e_a+ \Om^{2} |\nab R|^2 e_3$ and $e_3'=e_3$, $\rho'-\rho$ can be written as a linear combination of $\betab$ and $\alphab$. Thus, by \eqref{curv.zero}, in the limit $\de\to 0$,
$$\rho'\to \rho.$$
In the  $\de\to 0$ limit
 the Gauss equations \eqref{Gauss-eq2} becomes 
\beaa
\rho&=&-K     -\frac 14 \trch\trchb     =-R^{-2}+\frac{1}{2R}\trch
\eeaa
On the other hand, the  $\de\to 0$ limit
 of the Gauss equations \eqref{Gauss-eq1} becomes 
\beaa
K'=-\rho'-\frac 14 \trch'\trchb'=-\rho+\frac{1}{2 R}\trch'=     R^{-2}-\frac{1}{2R}\trch         +\frac{1}{2 R}\trch'
\eeaa
Thus, making use of \eqref{eq:Gauss-conf}, we have
\beaa
\trch'&=&\trch-\frac{1}{2R}+2R K'=\trch-\frac{2}{R}\De_0 (\log R)\\
&=&\trch+ \frac 2 R (-\frac{\De_0 R}{R}+\frac{|\nab_0 R|^2}{R^2}).
\eeaa
Combining this with Proposition \ref{prop:Christ}, we thus have
$$\trch'=\frac{2}{R}(1-\frac{\De_0 R}{R}+\frac{|\nab_0 R|^2}{R^2})-\frac{M_0}{R^2}$$
in the limit $\de\to 0$.
   
\section{Solutions to the deformation equation on $S_{0,0}$}\label{sec.solution}
To prove   our main     Theorem  \ref{Main:thm}  it suffices\footnote{Assuming    $\delta\delta_0^{-1}$, $\de_0$  sufficiently small, depending on $M_0$.  }   now to     show  that if $M_0$
verifies  the   the assumption \eqref{eq:main-thm2}, then an appropriate solution  to  the
differential inequality on the standard sphere  $S=S_{0,0}$,
\bea
 -\lap R+ R^{-1}|\nab R|^2+R<  2^{-1}  M_0, \label{trapped-Ineq2}
  \eea
  can be found.\footnote{From this point onwards, we drop the subscript $_0$ in the connection coefficients as it will be clear from context that we consider the connection coefficients associated to the initial sphere $S_{0,0}$.}

Let $R=e^{-\phi}$.
Then the main deformation equation \eqref{trapped-Ineq2} reduces to
\bea
\lap \phi+1< \frac 12 M_0 e^\phi\label{main.ineq}.
\eea
We show below  that (\ref{main.ineq}) can be solved 
as long as $M_0\geq 0 $ and $M_0\geq M_*>0$ on some open ball of  $S$. Our approach provides an explicit construction using the Green's function for the Laplacian on $S$. The main observation is that given any function $\tilde{\phi}$, there exists a sufficiently large constant $C$ such that \eqref{main.ineq} is satisfied by $\phi=\tilde{\phi}+C$ on the set where $M_0$ has a positive lower bound.
It is therefore sufficient first to construct a function $\tilde{\phi}$ satisfying \eqref{main.ineq} only on the complement of the set where $M_0$ has a positive lower bound. It turns out that an appropriately rescaled and cut-off version of the Green's function for the Laplacian satisfies this property.

We prove the following proposition, which together with Theorem \ref{thm5},
implies our main theorem (Theorems \ref{Main:thm} and \ref{main.thm.quan}):

\begin{proposition}\label{prop.Mep}
Let  $\displaystyle M_*=\min_{B_p(\ep)}  M_0$.
Then there exists a function $\phi_{\ep,M_*}$  verifying 
  the inequality \eqref{main.ineq}  and such that 
  \bea
  \phi_{\ep,M_*}& \leq & \log(\frac{1}{M_* \ep^5})+O(1)  \label{condphi1}\\
 | \nab\phi_{\ep,M_*}|&=&O(\ep^{-1}), \quad |\nab^2\phi_{\ep,M_*}|=O(\ep^{-2}).
  \eea
\end{proposition}

\subsection{Proof of the main theorem} Returning to our task  of constructing a trapped surface, 
 note that the upper bound  \eqref{condphi1}  for  $\phi$   corresponds  
  to our  desired   lower bound  for  $R$.
 More precisely, \eqref{condphi1} 
implies that for some $C>0$, 
$$\max_S \frac{1}{R}\leq e^{\max_S \phi}\leq \frac{C }{M_* \ep^5}.$$
In particular, for $M_0$ satisyfing the assumption of Theorem \ref{main.thm.quan}, i.e., 
$$\inf_{B_p(\ep)} M_0 \geq M_*> 0,$$
the proposition implies the existence of   a function  $R$ verifying \eqref{trapped-Ineq2}  and a lower bound $R> c_0{M_* \ep^5}$, for some constant $c_0$. 
Therefore, by Theorem \ref{Chr.existence.thm}, given $M_*$ and $\ep$, we can choose $\de$ sufficiently small such that the spacetime solution for the characteristic initial value problem remains smooth in $\DD(u_*,\de)$ for $u_*= 1-r_*<1-c_0{M_* \ep^5}$. This guarantees that the sphere $\{(\ub,u,\omega) : \ub=\delta, 1-u = R(\om)\}$ lies within $\DD(u_*,\de)$. Moreover, given a function $R$ verifying \eqref{trapped-Ineq2}, the term $C\de_0$ in \eqref{maineqR} can be made arbitrarily small by choosing $\de$ small. Thus, by Theorem \ref{thm5}, the sphere defined by $\{(\ub,u,\omega) : \ub=\delta, 1-u = R(\om)\}$ is a trapped surface in $\DD(u_*,\de)$. Since $R \geq r_*$, the constructed trapped surface has area at least $\gtrsim r_*^2 \approx M_*^2 \ep^{10}$.

This concludes the proof of the main theorem. It thus remains to prove Proposition    \ref{prop.Mep}.

\subsection{Proof of Proposition   \ref{prop.Mep}}
To this end,   we  use the following bounds for   the Green's function of  the standard  sphere $S$. 
\begin{lemma}\label{le:construct-w}
Given a point $p$ in the standard unit sphere, define $\lambda$ to be the distance function from $p$.
Then the function given by
$$w=\sin(\log(\frac{\lambda}{2}))$$
satisfies
\bea
\Delta_d w +\frac{1}{2 }= 2\pi \de_p  \label{weqn}
\eea
where $\Delta_d$ is the distributional Laplacian on the standard sphere
and $\de_p$ is the  Dirac measure at $p$.
Moreover, $w$ obey the following bounds:
$$w= \log \lambda+O(1),\quad |\nab w|=O(\la^{-1}),\quad |\nab^2 w|=O(\la^{-2}).$$           
\end{lemma}

\begin{remark}
In a more general setting where the metric $\ga$ on $S$ is not the metric on the standard unit sphere, we can still use the Green's function to construct a desired solution to the elliptic inequality. More precisely, given a smooth  riemannian  metric $\ga$  on $S$,
 there exists a function $w$, smooth outside the point $p$,   such that 
\bea
\Delta_d w +\frac{1}{2r_0 }= 2\pi \de_p  \label{weqn2}
\eea
where $\Delta_d$ is the distributional Laplacian  associated to  the metric $\ga$, 
$\de_p$ is the  Dirac measure at $p$, and $r_0$ is defined by $\mbox{Area}(S)= 4\pi r_0$.
Moreover, if $\la_p$ denotes the  distance function  from $p$,
\bea
w &=& \chi  \log \la_p  +       v
\eea
with      $v$  smooth in    $S\setminus\{p\} $     and satisfying
$$|v|\leq C,\quad |\nab v|=o(\la^{-1}),\quad |\nab^2 v|=o(\la^{-2});$$           
and $\chi $ a smooth cut-off function identically equals to 1 in a small neighborhood of $p$.

We refer the readers to Theorem 4.13 in \cite{Au} for a proof of this fact.
\end{remark}

Using Lemma \ref{le:construct-w}, we now proceed to the proof of Proposition \ref{prop.Mep}:

\begin{proof}[Proof of Proposition \ref{prop.Mep}]
Consider the cut-off function
 \beaa
\begin{cases}
&\chi_\ep=0\quad \mbox{on} \quad B_p(\ep/2)\\
&\chi_\ep=1\quad \mbox{on} \quad S\setminus B_p(\ep)
\end{cases}
\eeaa
and define $\tilde w_\ep=    \chi_\ep    w +  (1-\chi_\ep) \log \ep$.  Note\footnote{Note in particular  the logarithmic cancellation  in the formulas for
$\nab \tilde w_\ep$ and $\nab^2 \tilde w_\ep$.}
 for  that $\tilde w_\ep$  verifies the following
properties:
\bea
\begin{cases}
&\tilde w_\ep=\log\ep,\qquad\qquad\qquad \,\, \qquad \mbox{on}\qquad   B_p(\ep/2)\\
&\tilde w_\ep= \log\ep+O(1),\qquad\qquad  \mbox{on}\quad         B_p(\ep)\setminus       B_p(\ep/2)        \\
&\tilde w_\ep= \log \la +O(1),\qquad\qquad  \mbox{on}\quad         S \setminus       B_p(\ep)        \\
&\nab \tilde w_\ep= O(\ep^{-1}),\qquad  \,\, \qquad \mbox{on}\quad   
        S\setminus       B_p(\ep/2) \\
&\nab^2 \tilde w_\ep=  O(\ep^{-2}),\qquad  \,\, \qquad \mbox{on}\quad   
        S\setminus       B_p(\ep/2) \\
& \lap \tilde w_\ep +\frac{1}{2}=0,    \qquad \qquad  \quad   \qquad \mbox{on}\quad          S\setminus   B_p(\ep)
 \end{cases}
\eea
Consider now the function $ \tilde \phi_\ep=3 \tilde w_\ep$ and 
observe that, on $S\setminus   B_p(\ep)$, we must have,
$$ \lap \tilde \phi_\ep +1=-  \frac 32 +1<0.$$
Thus, we  have,
  \bea \label{phiepest}
\begin{cases}
&\tilde \phi_\ep= 3\log\ep,\qquad\qquad\qquad \,\, \qquad \mbox{on}\qquad   B_p(\ep/2)\\
&\tilde \phi_\ep= 3\log\ep+O(1),\qquad\qquad  \mbox{on}\quad         B_p(\ep)\setminus       B_p(\ep/2)        \\
&\tilde \phi_\ep= 3\log \la +O(1),\qquad\qquad  \mbox{on}\quad         S \setminus       B_p(\ep)        \\
&\nab \tilde \phi_\ep= O(\ep^{-1}),\qquad  \,\, \qquad \mbox{on}\quad   
        S\setminus       B_p(\ep/2) \\
&\nab^2 \tilde \phi_\ep=  O(\ep^{-2}),\qquad  \,\, \qquad \mbox{on}\quad   
        S\setminus       B_p(\ep/2) \\
& \lap \tilde \phi_\ep +1<0,    \qquad \qquad  \quad   \qquad \mbox{on}\quad          S\setminus   B_p(\ep)
 \end{cases}
\eea

Finally, we add a large constant to the function $\tilde{\phi_\ep}$ to obtain the desired solution to the inequality \eqref{main.ineq}. More precisely, we let 
\bea
\phi_{\ep, M_*}=-\log s_{\ep, M_*}+\tilde \phi_\ep,\label{phidef}
\eea
where $s_{\ep, M_*}$ is some small constant to be chosen later. 

Clearly, 
$$\lap \phi_{\ep,M_*}=\lap \tilde \phi_\ep.$$
Therefore, on $S\setminus B_p(\ep)$,
$$\lap \phi_{\ep,M_*}+1=\lap \tilde \phi_\ep +1<0.$$
By \eqref{phiepest}, there exists a positive constant $C'_0$ such that, everywhere,
$$\lap \tilde{\phi}_\ep\leq \frac{C'_0}{\ep^2}.$$
Thus, on $B_p(\ep)$,
$$\lap \phi_{\ep,M_*}+1=\lap \tilde \phi_\ep +1<\frac{C_0}{\ep^{2}}$$
for some positive constant $C_0>0$.      To ensure our inequality  \eqref{main.ineq}   we  need to make sure that  $\frac{C_0}{\ep^{2}}$    does not exceed $\frac 12 M_* e^{\phi_{\ep,M_*}}$ on $B_p(\ep)$.    We therefore need,
$$\inf_{B_p(\ep)} \phi_{\ep,M_*} \geq \log(\frac{2 C_0}{M_*\ep^2}).$$
Recall that \eqref{phiepest} implies
$$\inf_{B_p(\ep)} \phi_{\ep,M_*}\geq -\log s_{\ep, M_*}+3\log \ep-\log C_1$$
for some $C_1>0$. Therefore, it suffices to choose
\bea
s_{\ep, M_*}=\frac{M_* \ep^{5}}{2 C_0 C_1} \label{phidef2}
\eea
and \eqref{main.ineq} is verified everywhere on $S$.

Finally, we use \eqref{phiepest}, \eqref{phidef} and \eqref{phidef2} to check that $\phi_{\ep, M_*}$ obeys the bounds asserted in the proposition. By \eqref{phiepest}, we have the one-sided bound
$$\phi_{\ep} \leq C,$$
for some $C>0$. This, together with \eqref{phidef} and \eqref{phidef2} implies
\bea
\sup_S \phi_{\ep,M_*} \leq -\log (\frac{M_* \ep^{5}}{C_0 C_1})+O(1) = \log(\frac{1}{M_* \ep^5})+O(1).
\eea
For the bounds for the first and second derivatives of $\phi_{\ep,M_*}$, notice that since $s_{\ep,M_*}$ is a constant, by \eqref{phiepest}, we have
\beaa
|\nab\phi_{\ep,M_*}|&=&|\nab\phi_{\ep}|=O(\ep^{-1}),\\
|\nab^2\phi_{\ep,M_*}|&=&|\nab^2\phi_{\ep}|=O(\ep^{-2}).
\eeaa
\end{proof}

\section{Acknowledgements}    
We would like to thank Demetrios Christodoulou for   his   helpful  suggestions      on a previous version of the manuscript.
We would also like to thank   Pin Yu,     Ovidiu Savin and Stefanos Aretakis for    helpful    remarks and  useful   discussions. 

S Klainerman would like to thank the Fondation SMP for the support of this work. 
S. Klainerman is supported by    the   NSF  grant   DMS-0901250.
J. Luk is supported by the NSF Postdoctoral Fellowship DMS-1204493.
I. Rodnianski  is supported by the NSF grant DMS-1001500. S. Klainerman and I. Rodnianski are
also  supported  by the FRG grant  DMS-1065710.

\end{document}